\documentclass[11pt]{article}

\usepackage[margin=2cm]{geometry}
\usepackage[utf8]{inputenc}
\usepackage{amsmath}
\usepackage{amssymb}
\usepackage{amsthm}
\usepackage{color} 
\usepackage{dsfont}
\usepackage{wasysym}
\usepackage{graphicx}
\usepackage{algpseudocode}
\usepackage{hyperref}

\newtheorem{definition}{Definition}
\newtheorem{lemma}{Lemma} 
\newtheorem{theorem}{Theorem}

\newtheorem{observation}{Observation}

\newcommand{\N}{\mathds{N}}

\newcommand{\Polytime}{\mathrm{P}}
\newcommand{\NP}{\mathrm{NP}}
\newcommand{\coNP}{\mathrm{coNP}}
\newcommand{\UP}{\mathrm{UP}}
\newcommand{\PLS}{\mathrm{PLS}}
\newcommand{\coUP}{\mathrm{coUP}}
\newcommand{\scurr}{\mathop{{\tt s\_curr}}}
\newcommand{\snext}{\mathop{{\tt s\_next}}}
\newcommand{\Arrival}{\mbox{\rm ARRIVAL}}
\newcommand{\Run}{\mbox{\sc Run}}
\newcommand{\xx}{\mathbf{x}}

\title{$\Arrival$: A zero-player graph game in  $\NP\cap\coNP$\thanks{This research
    was done in the 2014 undergraduate seminar \emph{Wie
      funktioniert Forschung}? (How does research work?)}}
\author{J\'{e}r\^{o}me Dohrau, Bernd G\"artner, Manuel Kohler, Ji\v{r}\'{i} Matou\v{s}ek,
  Emo Welzl\thanks{Department of Computer Science, Institute of Theoretical
Computer Science, ETH Z\"urich, CH-8092 Z\"urich, Switzerland, {\tt
  gaertner@inf.ethz.ch} (corresponding author)}}

\begin{document}
\maketitle

\begin{abstract}
  Suppose that a train is running along a railway network, starting
  from a designated origin, with the goal of reaching a designated
  destination. The network, however, is of a special nature: every
  time the train traverses a switch, the switch will change its
  position immediately afterwards. Hence, the next time the train
  traverses the same switch, the other direction will be taken, so
  that directions alternate with each traversal of the switch.

  Given a network with origin and destination, what is the
  complexity of deciding whether the train, starting at the origin,
  will eventually reach the destination?

  It is easy to see that this problem can be solved in exponential time,
  but we are not aware of any polynomial-time method. In this short
  paper, we prove that the problem is in $\NP\cap \coNP$. This raises
  the question whether we have just failed to find a (simple)
  polynomial-time solution, or whether the complexity status is more
  subtle, as for some other well-known (two-player) graph
  games~\cite{Halman}.
\end{abstract}

\section{Introduction}
In this paper, a \emph{switch graph} is a directed graph $G$ in which
every vertex has at most two outgoing edges, pointing to its \emph{even}
and to its \emph{odd} successor. Formally, a switch graph is a
4-tuple $G=(V,E,s_0,s_1)$, where $s_0, s_1:V\rightarrow V$, $E=
\{(v,s_0(v)): v\in V\} \cup \{(v,s_1(v)): v\in V\}$, with loops
$(v,v)$ allowed.  Here, $s_0(v)$ is the even successor of $v$, and
$s_1(v)$ the odd successor.  We may have $s_0(v)=s_1(v)$ in which
case $v$ has just one outgoing edge. We always let $n=|V|$; for $v\in
V$, $E^+(v)$ denotes the set of outgoing edges at $v$, while $E^-(v)$
is the set of incoming edges.

Given a switch graph $G=(V,E,s_0,s_1)$ with origin and destination
$o,d\in V$, the following procedure describes the train run that we
want to analyze; our problem is to decide whether the procedure
terminates. For the procedure, we assume arrays $\scurr$ and $\snext$,
indexed by $V$, such that initially $\scurr[v]=s_0(v)$ and
$\snext[v]=s_1(v)$ for all $v\in V$.\medskip

\begin{algorithmic}
\Procedure{Run}{$G,o,d$}
\State{$v := o$}
\While{$v\neq d$} 
\State{$w := \scurr[v]$}
\State{swap ($\scurr[v], \snext[v]$)}
\State{$v := w$} \Comment{traverse edge $(v,w)$}
\EndWhile 
\EndProcedure
\end{algorithmic}

\begin{definition}
  Problem $\Arrival$ is to decide whether procedure
  $\Run(G,o,d)$ terminates for a given switch graph $G=(V,E,s_0,s_1)$
  and $o,d\in V$.
\end{definition}

\begin{theorem}\label{thm:decidable}
Problem $\Arrival$ is decidable.
\end{theorem}

\begin{proof}
  The deterministic procedure $\Run$ can be interpreted as a function
  that maps the current \emph{state} $(v, \scurr,\snext)$ to the next
  state. We can think of the state as the current location of the
  train, and the current positions of all the switches. As at most $n2^n$
  different states may occur, $\Run$ either terminates within this
  many iterations, or some state repeats, in which case $\Run$ enters
  an infinite loop. Hence, to decide $\Arrival$, we have to go through
  at most $n2^n$ iterations of $\Run$.
\end{proof}

Figure~\ref{fig:exponential} shows that a terminating run may indeed
take exponential time.

\begin{figure}[htb]
\begin{center}
\includegraphics[width=0.5\textwidth]{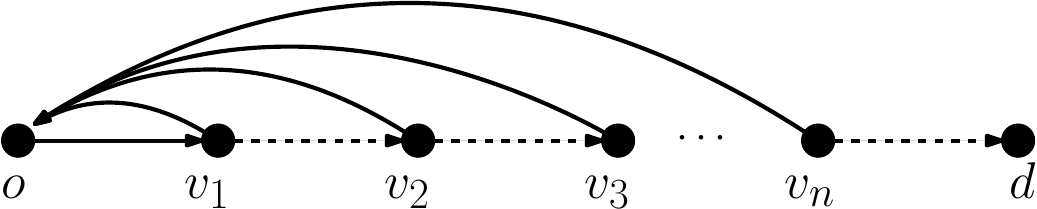}
\end{center}
\caption{Switch graph $G$ with $n+2$ vertices on which $\Run(G,o,d)$
  traverses an exponential number of edges. If we encode the
  current positions of the switches at $v_n,\ldots,v_1$ with an
  $n$-bit binary number (0: even successor is next; 1: odd successor is
  next), then the run counts from $0$ to $2^n-1$, resets the
  counter to $0$, and terminates. Solid edges point to even or unique successors, dashed
  edges to odd successors.}\label{fig:exponential}
\end{figure}

Existing research on switch graphs (with the above, or similar
definitions) has mostly focused on actively controlling the switches,
with the goal of attaining some desired behavior of the network
(e.g.\ reachability of the destination); see e.g.~\cite{Katz2012}. The
question that we address here rather fits into the theory of cellular
automata. It is motivated by the online game \emph{Looping Piggy}
(\url{https://scratch.mit.edu/projects/1200078/}) that the second
author has written for the \emph{Kinderlabor}, a Swiss-based
initiative to educate children at the ages 4--12 in natural sciences
and computer science (\url{http://kinderlabor.ch}).

It was shown by Chalcraft and Greene~\cite{CG} (see also
Stewart~\cite{Stu}) that the train run can be made to simulate a
Turing machine if on top of our ``flip-flop switches'', two other
types of natural switches can be used. Consequently, the arrival
problem is undecidable in this richer model; it then also becomes
NP-complete to decide whether the train reaches the destination for
\emph{some} initial positions of a set of flip-flop
switches~\cite{Mar}.

Restricting to flip-flop switches with fixed initial positions, the
situation is much less complex, as we show in this paper. In
Sections~\ref{sec:np} and \ref{sec:conp}, we prove that $\Arrival$ is
in $\NP$ as well as in $\coNP$; Section~\ref{sec:ip} shows that a
terminating run can be interpreted as the unique solution of a
flow-type integer program with balancing conditions whose LP
relaxation may have only fractional optimal solutions.

\section{$\Arrival$ is in $\NP$}\label{sec:np}
A natural candidate for an $\NP$-certificate is the \emph{run profile}
of a terminating run. The run profile assigns to each edge the number
of times it has been traversed during the run. The main difficulty is
to show that fake run profiles cannot fool the verifier. We start with
a necessary condition for a run profile: it has to be a
\emph{switching flow}. 

\begin{definition}
  Let $G=(V,E,s_0,s_1)$ be a switch graph, and let $o,d\in V$, $o\neq
  d$. A \emph{switching flow} is a function
  $\xx:E\rightarrow\N_0$ (where $\xx(e)$ is denoted as $x_e$) such that the
  following two conditions hold for all $v\in V$.
\begin{eqnarray}
\sum_{e\in E^+(v)} x_e - \sum_{e\in E^-(v)} x_e =
\left\{
  \begin{array}{rl}
    1, & v = o, \\
  -1, & v = d, \\
   0, & \mbox{otherwise}.
  \end{array}
\right.
\label{eq:conserve_flow}
\end{eqnarray}

\begin{eqnarray}
0 \leq x_{(v, s_1(v))} \leq x_{(v, s_0(v))} \leq x_{(v, s_1(v))} + 1. 
\label{eq:balance_flow}
\end{eqnarray}
\end{definition}

\begin{observation}\label{obs:termination_flow}
  Let $G=(V,E,s_0,s_1)$ be a switch graph, and let $o,d\in V$, $o\neq
  d$, such that $\Run(G,o,d)$ terminates. Let $\xx(G,o,d):E\rightarrow
  \N_0$ (the run profile) be the function that assigns to each edge
  the number of times it has been traversed during $\Run(G,o,d)$. Then
  $\xx(G,o,d)$ is a switching flow.
\end{observation}

\begin{proof}
Condition (\ref{eq:conserve_flow}) is simply flow conservation (if the
run enters a vertex, it has to leave it, except at $o$ and $d$), 
while (\ref{eq:balance_flow}) follows from the run alternating
between successors at any vertex $v$, with the even successor $s_0(v)$
being first.
\end{proof}

\begin{figure}[htb]
\begin{center}
\includegraphics[width=0.3\textwidth]{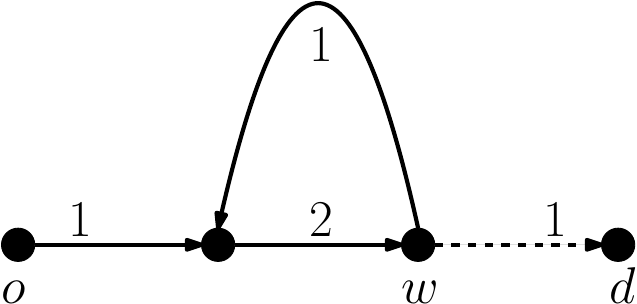}~~~~~~~~\includegraphics[width=0.3\textwidth]{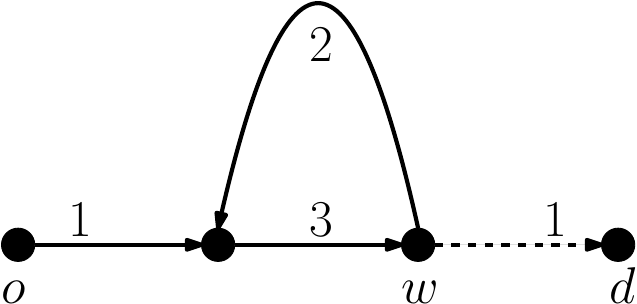} 
\end{center}
\caption{Run profile (left) and fake run profile (right); both are
  switching flows. Solid edges point to even or unique successors,
  dashed edges to odd successors.}\label{fig:fakerun}
\end{figure}

While every run profile is a switching flow, the converse is not
always true. Figure~\ref{fig:fakerun} shows two switching flows for
the same switch graph, but only one of them is the actual run
profile. The ``fake'' run results from going to the even
successor of $w$ twice in a row, before going to the odd successor
$d$. This shows that the balancing condition (\ref{eq:balance_flow})
fails to capture the strict alternation between even and odd
successors. Despite this, and maybe surprisingly, the existence of a switching flow
implies termination of the run.

\begin{lemma}\label{lem:flow_termination}
 Let $G=(V,E,s_0,s_1)$ be a switch graph, and let $o,d\in V$, $o\neq
  d$. If there exists a switching flow $\xx$, then
  $\Run(G,o,d)$ terminates, and $\xx(G,o,d)\leq\xx$ (componentwise).
\end{lemma}

\begin{proof}
  We imagine that for all $e\in E$ we put $x_e$ pebbles on edge $e$,
  and then start $\Run(G,o,d)$. Every time an edge is traversed, we let
  the run collect one pebble. The claim is that we
  never run out of pebbles, which proves termination as well as the
  inequality for the run profile.

  To prove the claim, we first observe two invariants: during the run, flow
  conservation (w.r.t.\ to the remaining pebbles) always holds, except
  at $d$, and at the current vertex which has one more pebble on its
  outgoing edges. Moreover, by alternation, starting with the even
  successor, the numbers of pebbles on $(v,s_0(v))$ and $(v,s_1(v))$
  always differ by at most one, for every vertex $v$.

  For contradiction, consider now the first iteration of $\Run(G,o,d)$
  where we run out of pebbles, and let $e=(v,w)$ be the edge (now
  holding $-1$ pebbles) traversed in the offending iteration. By the
  above alternation invariant, the other outgoing edge at $v$ cannot
  have any pebbles left, either. Then the flow conservation invariant
  at $v$ shows that already some incoming edge of $v$ has a deficit of
  pebbles, so we have run out of pebbles before, which is a
  contradiction.
\end{proof}

\begin{theorem}\label{thm:np}
Problem $\Arrival$ is in $\NP$.
\end{theorem}

\begin{proof}
  Given an instance $(G,o, d)$, the verifier receives a function
  $\xx:E\rightarrow\N_0$, in form of binary encodings of the values
  $x_e$, and checks whether it is a switching flow. For a
  Yes-instance, the run profile of $\Run(G,o,d)$ is a witness by
  Observation~\ref{obs:termination_flow}; the
  proof of Theorem~\ref{thm:decidable} implies that the verification
  can be made to run in polynomial time, since every value $x_e$ is
  bounded by $n2^n$. For a No-instance, the check will fail by
  Lemma~\ref{lem:flow_termination}.
\end{proof}

\section{$\Arrival$ is in $\coNP$}\label{sec:conp}
Given an instance $(G,o,d)$ of $\Arrival$, the main idea is to
construct in polynomial time an instance $(\bar{G},o,\bar{d})$ such
that $\Run(G,o,d)$ terminates if and only if
$\Run(\bar{G},o,\bar{d})$ does not terminate. As the main
technical tool, we prove that nontermination is equivalent to the
arrival at a ``dead end''.

\begin{definition}
  Let $G=(V,E,s_0,s_1)$ be a switch graph, and let $o,d\in V$, $o\neq
  d$. A \emph{dead end} is a vertex from which there is no directed
  path to the destination $d$ in the graph $(V,E)$.  A \emph{dead
    edge} is an edge $e=(v,w)$ whose head $w$ is a dead end. An edge
  that is not dead is called \emph{hopeful}; the length of the
  shortest directed path from its head $w$ to $d$ is called its
  \emph{desperation}.
\end{definition}

By computing the tree of shortest paths to $d$, using inverse
breadth-first search from $d$, we can identify the dead ends in
polynomial time. Obviously, if $\Run(G,o,d)$ ever reaches a dead
end, it will not terminate, but the converse is also true. For this,
we need one auxiliary result.

\begin{lemma}\label{lem:glimmer}
  Let $G=(V,E,s_0,s_1)$ be a switch graph, $o,d\in V$, $o\neq d$, and
  let $e=(v,w)\in E$ be a hopeful edge of desperation $k$. Then
  $\Run(G,o,d)$ will traverse $e$ at most $2^{k+1}-1$ times.
\end{lemma}

\begin{proof} Induction on the desperation $k$ of $e=(v,w)$. If $k=0$,
  then $w=d$, and indeed, the run will traverse $e$ at most $2^1-1=1$
  times. Now suppose $k>0$ and assume that the statement is true for
  all hopeful edges of desperation $k-1$. In particular, one of the
  two successor edges $(w, s_0(w))$ and $(w,s_1(w))$ is such a hopeful
  edge, and is therefore traversed at most $2^{k}-1$ times. By
  alternation at $w$, the other successor edge is traversed at most
  once more, hence at most $2^k$ times. By flow conservation, the
  edges entering $w$ (in particular $e$) can be traversed at most
  $2^k+2^k-1=2^{k+1}-1$ times.
\end{proof}

\begin{lemma}\label{lem:deadend}
  Let $G=(V,E,s_0,s_1)$ be a switch graph, and let $o,d\in V$, $o\neq
  d$. If $\Run(G,o,d)$ does not terminate, it will reach a dead end.
\end{lemma}

\begin{proof} By Lemma~\ref{lem:glimmer}, hopeful edges can be
  traversed only finitely many times, hence if the run cycles, it
  eventually has to traverse a dead edge and thus reach a dead end.
\end{proof}

Now we can prove the main result of this section.

\begin{theorem}
Problem $\Arrival$ is in $\coNP$.
\end{theorem}

\begin{proof}
  Let $(G,o,d)$ be an instance, $G=(V,E,s_0,s_1)$. We transform
  $(G,o,d)$ into a
  new instance $(\bar{G},o,\bar{d})$,
  $\bar{G}=(\bar{V},\bar{E},\bar{s}_0,\bar{s}_1)$ as follows. We set
  $\bar{V}=V\cup\{\bar{d}\}$, where $\bar{d}$ is an additional vertex, the
  new destination. We define $\bar{s}_0,\bar{s}_1$ as follows. For
  every dead end $w$, we set
\begin{equation}
\bar{s}_0 (w) = \bar{s}_1 (w) := \bar{d}.
\label{eq:newdest}
\end{equation}
For the old destination $d$, we install the loop
\begin{equation}
\bar{s}_0 (d) = \bar{s}_1 (d) := d.
\label{eq:loop}
\end{equation}
For the new destination, $\bar{s}_0(\bar{d})$ and
$\bar{s}_1(\bar{d})$ are chosen arbitrarily. In all other cases,
$\bar{s}_0(v):=s_0(v)$ and $\bar{s}_1(v):=s_1(v)$. This defines
$\bar{E}$ and hence $\bar{G}$.

The crucial properties of this construction are the following: 

\begin{itemize}
\item[(i)] If $\Run(G,o,d)$ reaches the destination $d$, it has not
  visited any dead ends, hence $s_0$ and $\bar{s}_0$ as well as $s_1$
  and $\bar{s}_1$ agree on all visited vertices except $d$. This means
  that $\Run(\bar{G},o,\bar{d})$ will also reach $d$, but then cycle
  due to the loop that we have installed in \eqref{eq:loop}.
\item[(ii)] If $\Run(G,o,d)$ cycles, it will at some point reach a first
  dead end $w$, by Lemma~\ref{lem:deadend}. As $s_0$ and $\bar{s}_0$ as
  well as $s_1$ and $\bar{s}_1$ agree on all previously visited
  vertices, $\Run(\bar{G},o,\bar{d})$ will also reach $w$, but then
  terminate due to the edges from $w$ to $\bar{d}$ that we have
  installed in \eqref{eq:newdest}.
\end{itemize}

To summarize, $\Run(G,o,d)$ terminates if and only if
$\Run(\bar{G},o,\bar{d})$ does not terminate. Since
$(\bar{G},o,\bar{d})$ can be constructed in polynomial time, we can
verify in polynomial time that $(G,o,d)$ is a No-instance by verifying
that $(\bar{G},o,\bar{d})$ is a Yes-Instance via Theorem~\ref{thm:np}.
\end{proof}

\section{Is $\Arrival$ in $\Polytime$?} \label{sec:ip}
Observation~\ref{obs:termination_flow} and
Lemma~\ref{lem:flow_termination} show that $\Arrival$ can be decided
by checking the solvability of a system of linear (in)equalities
(\ref{eq:conserve_flow}) and (\ref{eq:balance_flow}) over the
nonnegative integers.

The latter is an NP-complete problem in general: many of the
standard NP-complete problems, e.g.\ SAT (satisfiability of
boolean formulas) can easily be reduced to finding an integral vector
that satisfies a system of linear (in)equalities.

In our case, we have a flow structure, though, and finding integral
flows in a network is a well-studied and easy problem~\cite[Chapter
8]{KV12}. In particular, if only the flow conservation constraints
(\ref{eq:conserve_flow}) are taken into account, the existence of a
nonnegative integral solution is equivalent to the existence of a
nonnegative real solution. This follows from the classical
\emph{Integral Flow Theorem}, see~\cite[Corollary 8.7]{KV12}. Real
solutions to systems of linear (in)equalities can be found in
polynomial time through linear programming~\cite[Chapter 4]{KV12}.

However, the additional balancing constraints (\ref{eq:balance_flow})
induced by alternation at the switches, make the situation more
complicated. Figure~\ref{fig:fractional} depicts an instance which has
a real-valued ``switching flow'' satisfying constraints
(\ref{eq:conserve_flow}) and (\ref{eq:balance_flow}), but no integral
one (since the run does not terminate). 

\begin{figure}[htb]
\begin{center}
\includegraphics[width=0.4\textwidth]{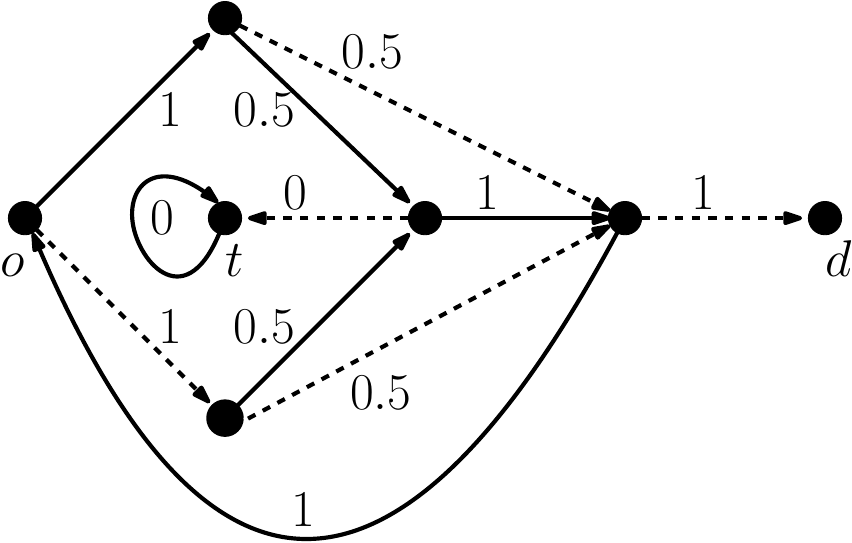}
\end{center}
\caption{The  run will enter the loop at $t$ and cycle, so there is no
  (integral) switching flow.  But a real-valued ``switching flow''
  (given by the numbers) exists. Solid edges point to even or unique successors,
  dashed edges to odd successors. \label{fig:fractional}}
\end{figure}

We conclude with a result that summarizes the situation and may be
the basis for further investigations.

\begin{theorem}\label{thm:ip}
  Let $G=(V,E,s_0,s_1)$ be a switch graph, and let $o,d\in V$, $o\neq
  d$. $\Run(G,o,d)$ terminates if and only if there exists an integral
  solution satisfying the constraints (\ref{eq:conserve_flow}) and
  (\ref{eq:balance_flow}). In this case, the run profile $\xx(G,o,d)$
  is the unique integral solution that minimizes the linear objective
  function $\Sigma(\xx)=\sum_{e\in E}x_e$ subject to the constraints
  (\ref{eq:conserve_flow}) and (\ref{eq:balance_flow}). 
\end{theorem}

\begin{proof}
  Observation~\ref{obs:termination_flow} and
  Lemma~\ref{lem:flow_termination} show the equivalence between
  termination and existence of an integral solution (a switching
  flow). Suppose that the run terminates with run profile
  $\xx(G,o,d)$. We have $\xx(G,o,d)\leq\xx$ for every switching flow
  $\xx$, by Lemma~\ref{lem:flow_termination}. In particular,
  $\Sigma(\xx(G,o,d))\leq\Sigma(\xx)$, so the run profile has minimum
  value among all switching flows. A different switching flow $\xx$
  of the same value would have to be smaller in at least one
  coordinate, contradicting $\xx(G,o,d)\leq\xx$.
\end{proof}

Theorem~\ref{thm:ip} shows that the existence of $\xx(G,o,d)$ and its
value can be established by solving an integer program~\cite[Chapter
5]{KV12}. Moreover, this integer program is of a special kind: its
unique optimal solution is at the same time a least element w.r.t.\ the
partial order ``$\leq$'' over the set of feasible solutions.

\section{Conclusion}
The main question left open is whether the zero-player graph game
$\Arrival$ is in $\Polytime$. There are three well-known two-player
graph games in $\NP\cap\coNP$ for which membership in $\Polytime$ is
also not established: \emph{simple stochastic games}, \emph{parity
  games}, and \emph{mean-payoff games}. All three are
even in $\UP\cap\coUP$, meaning that there exist efficient verifiers
for Yes- and No-instances that accept \emph{unique}
certificates~\cite{Con,Jur}. In all three cases, the way to prove this
is to assign payoffs to the vertices in such a way that they form a
certificate if and only if they solve a system of equations with a
unique solution.

It is natural to ask whether also $\Arrival$ is in $\UP\cap\coUP$. We
do not know the answer. The natural approach suggested by
Theorem~\ref{thm:ip} is to come up with a verifier that does not
accept just any switching flow, but only the unique one of minimum
norm corresponding to the run profile. However, verifying optimality
of a feasible integer program solution is hard in general, so for this
approach to work, one would have to exploit specific structure of the
integer program at hand. We do not know how to do this.

As problems in $\NP\cap\coNP$ cannot be $\NP$-hard (unless $\NP$ and
$\coNP$ collapse), other concepts of hardness could be considered for
$\Arrival$. As a first step in this direction, Karthik C.~S.~\cite{Kar} has
shown that a natural search version of $\Arrival$ is contained in the
complexity class $\PLS$ (Polynomial Local Search) which has complete
problems not known to be solvable in polynomial time. $\PLS$-hardness
of $\Arrival$ would not contradict common complexity theoretic
beliefs; establishing such a hardness result would at least provide a
satisfactory explanation why we have not been able to find a
polynomial-time algorithm for $\Arrival$.

\section{Acknowledgment} 
We thank the referees for valuable comments and Rico
Zenklusen for constructive discussions. 

\bibliography{piggy_arxiv}
\bibliographystyle{plain}

\end{document}